\documentclass[preprint,11pt]{elsarticle}
\usepackage{amsfonts,amssymb,amsmath,latexsym,ae,aecompl}
\usepackage{epsfig}
\usepackage{graphicx}
\usepackage{algorithmic}
\usepackage{psfrag}
\usepackage{pstricks, pst-all, pstricks-add}
\usepackage{epsfig}
\usepackage{float,epsfig, floatflt,here}
\usepackage{subfig}
\usepackage{a4wide}

\newtheorem{theorem}{Theorem}
\newtheorem{corollary}{Corollary}
\newtheorem{definition}[theorem]{Definition}

\newenvironment{proof}{\noindent{\bf Proof: }}{\qed \smallbreak}

\usepackage[ruled,vlined]{algorithm2e}

\begin{document}
\begin{frontmatter}

\title{Designing Path Planning Algorithms for\\ Mobile Anchor towards Range-Free Localization}
\author[label1]{Kaushik Mondal}
\author[label2]{Arindam Karmakar}
\author[label1]{Partha Sarathi Mandal}
\address[label1]{Indian Institute of Technology Guwahati, India}
\address[label2]{Tezpur University, India}


\begin{abstract}
Localization is one of the most important factor in wireless sensor networks as
many applications demand position information of sensors.
Recently there is an increasing interest on the use of mobile anchors for localizing sensors.
Most of the works available in the literature either looks into the aspect of reducing path
length of mobile anchor or tries to increase localization accuracy.
The challenge is to design a movement strategy for a mobile anchor that reduces path length
while meeting the requirements of a good range-free localization technique.
In this paper we propose two cost-effective movement strategies i.e., path planning for a mobile anchor
so that localization can be done using the localization scheme \cite{Lee2009}.
In one strategy we use a hexagonal movement pattern for the mobile anchor to localize all sensors inside a bounded
rectangular region with lesser movement compared to the existing works in literature.
In other strategy we consider a connected network in an unbounded region where the mobile anchor
moves in the hexagonal pattern to localize the sensors. In this approach, we guarantee localization of
all sensors within $r/2$ error-bound where $r$ is the communication range of the mobile anchor and sensors.
Our simulation results support theoretical results along with localization accuracy.
\end{abstract}

\begin{keyword}
Localization\sep Range-free\sep Beacon point\sep Path planning\sep Mobile anchor\sep Wireless Sensor Networks
\end{keyword}

\end{frontmatter}

\section{Introduction}
\label{sec:intro}
Localization of wireless sensors with high degree of accuracy is required for many
wireless sensor networks (WSNs) applications, such as security and surveillance,
object tracking, detecting accurate location of a target etc.
To meet this purpose, many sensor localization schemes
\cite{Ammar2010,Chen2012,Delaet11,Lee2009,Seow:2008,Ssu2005,Xiao2008,zhang2006} have been proposed for WSNs.
These schemes can be viewed as range-based or range-free. In range-based schemes, the sensor locations
are calculated using distance and/or angle information among sensors by ranging hardware.
On the other hand, connectivity constraints such as hop-count, anchor beacons etc are used in range-free
schemes. Usually range-based schemes are more accurate than range-free schemes. But range estimation
techniques in range-based schemes are erroneous as well as costly due to requirement of special hardware,
which encouraged researcher to design range-free schemes for sensor localization.
A number of static anchors are needed in the localization schemes like \cite{Ammar2010,Chen2012,Delaet11,Seow:2008,zhang2006}
which uses static anchors.
To minimize the number of static anchors, localization schemes \cite{Lee2009,Ssu2005,Xiao2008} using mobile anchor are proposed.
One mobile anchor with a suitable path planning is equivalent to many static anchors,
which localizes whole network. By using mobile anchor, we can save large number of anchors with deployment cost
in the expense of the mobility of the mobile anchor. So, path planning of the mobile anchor has
become an important issue in the area of localization. 
There are few proposed movement strategies which localizes sensors using some basic
techniques like trilateration \cite{Shih2010}, which causes large localization error. Our aim is to
propose a movement strategy such that we can use existing range-free localization schemes which
yields better accuracy.

In this paper we have proposed path planning schemes for the mobile anchor where localization
is done using localization scheme proposed by Lee et al.\cite{Lee2009}.
We have proposed two different movement strategies for two different assumptions.
One movement strategy is proposed on the assumption that the network is connected.
The anchor localizes every sensor of the network with connectivity guided movement.
The other movement strategy is proposed to localize sensors over a bounded rectangular
region where the anchor has to cover the whole rectangle to ensure localization of all sensors.
In this strategy, only boundary information is used to localize all the sensors irrespective of
deployment and underlying network topology.

\subsection*{Our contribution:}
In this paper we have proposed a hexagonal movement strategy for a mobile anchor
to localize static sensors with improved (lesser) movement of the mobile anchor
compared to the existing results in literature. We have divided our work in two parts.
First part assumes connectivity in the network whereas later we have used known boundary
of a rectangular region where the static sensors reside. Our achievements are the following.
\begin{itemize}
\item We have proposed a distributed range-free movement strategy to localize all sensors within $r/2$ error-bound
in a connected network, where $r$ is the transmission range of the sensors and the mobile anchor.
\item We have given another path planning scheme for a bounded rectangular region
using same hexagonal movement pattern.
\item Theoretically we have shown that
the length of the path traversed by the anchor is lesser in the
proposed strategy compared to other existing path planning methods for covering
a rectangular region. 
\item Our simulation results support all theoretical results for path planning with localization accuracy.
Simulation results show $13.45 \%$ to $25.35 \%$ on an average improvement of our scheme over different schemes in
terms of path length while covering a bounded rectangular region.
\end{itemize}

The organization of the remaining part of the paper is as follows. In section \ref{sec:reltdwrk} we discuss
about related works. The theoretical results of our
proposed path planning on a connected network are explained in section \ref{sec:tech1}.
The algorithm along with system model are given in section \ref{sec:algomodel}.
In section \ref{sec:tech2}, we propose a movement strategy of mobile anchor for a bounded rectangular region.
The simulation results are presented in section \ref{sec:sim}, along with performance
comparison with existing approaches. Finally we conclude in section \ref{sec:conclusion}.

\section{Related Works}
\label{sec:reltdwrk}
Path planning algorithms set path for mobile anchor along which it moves in
the network while localization process goes on. First we look at a
brief overview of the existing range free localization schemes which provides
good accuracy and can be used for localization. Ssu et al. proposed
a localization scheme in \cite{Ssu2005} where the sensor's position is estimated as the
intersection of perpendicular bisector of two calculated chords.
However this scheme suffers from short chord length problem. Xiao et al.\cite{Xiao2008} improved
over that scheme using pre-arrival and post-departure points along with the beacon points
to localize a sensor. Later Lee et al. used beacon distance more effectively as another
geometric constraint and proposed a more accurate localization scheme in \cite{Lee2009}.

We can view the path planning problem in two different ways depending on the knowledge of the area of sensor deployment
and the underlying topology formed by the sensors. Topology-based path planning, can be viewed as a
graph traversal problem. Sensors have information about their neighbors which they send to the mobile anchor
for determining the path. Li et al. proposed two algorithms named breadth first
and backtracking greedy algorithms in \cite{Li2008}.
Mitton et al. in \cite{Mitton2012} proposed a depth first traversal scheme by the mobile anchor
to localize the sensors. Both these works need range estimations.
Kim et al. proposed a path planning in \cite{Kim2011} for  randomly deployed sensors
using trilateration method for localization. An already localized sensor
becomes a reference point to help other sensors to find their position
which reduces path length but localization error may propagate.
Chang  et al. proposed another path planning algorithm of the mobile anchor
in \cite{Chang2012} where localization have been done using the scheme proposed by Galstyan et al.
in \cite{Galstyan2004} and mobile sensor calculates its trajectory by moving around already localized sensors.
Our aim is to propose a path planning algorithm which can decide its trajectory without using any range
estimation in a connected network. Using connectivity of the network, we discover neighbors of a sensor as
well as localize it by the scheme \cite{Lee2009} using our proposed path planning algorithm.

The other way of viewing the path planning problem is to cover a rectangular area
by the mobile anchor where all the sensors are deployed.
$Scan$, $Doublescan$ and $Hilbert$ schemes are proposed by Koutsonikolas et al. in \cite{Koutsonikolas2007}.
They used the localization scheme proposed in \cite{Sichitiu2004}. Scan covers the whole area
uniformly where the mobile anchor travels in line segments along $x$-axis (or $y$-axis) keeping a
fixed distance between two line segments. In $Doublescan$, anchor moves along both $x$-axis
and $y$-axis, which improves localization accuracy in the expense of traveled distance.
$Hilbert$ reduces both error and path length with compare to the other two.
Huang et al. proposed two path planning schemes namely $Circles$ and {\it S-curves} in \cite{Huang2007}.
Simulation results show that these two schemes produce better results than those discussed above.
Based on trilateration, Han et al.  proposed a path planning scheme in \cite{Han2013} for a
mobile anchor. Using Received Signal Strength Indicator (RSSI) technique, sensor measures distances from three
different non collinear points and finds its position. Chia-Ho-Ou et al. proposed
a movement strategy in \cite{Chia-Ho-Ou2013} of the anchor which helps sensors to localize
with good accuracy by reducing the short chord length problem of Ssu's scheme \cite{Ssu2005}.
Our aim is to propose a path planning which minimizes the path length compared to the
existing ones and guarantee positioning of each sensor using scheme
proposed by Lee et al. in \cite{Lee2009}.

\section{Path Planning for Connected Network}
\label{sec:tech1}
In this section we discuss path planning to localize an arbitrary connected network of any number of sensors.
The mobile anchor broadcasts beacon with its position information after every $t$ time interval.
We may use the term `anchor' instead `mobile anchor' in the rest part of this paper. Here required definitions are given below.
\begin{definition}
{\rm(}Beacon distance{\rm)} Distance traveled by the mobile anchor between two consecutive broadcasts
of beacon is called {\it beacon distance} and is denoted by $u$.
\end{definition}
\begin{definition}
{\rm(}Communication circle{\rm)} The circle with radius $r$ centering at the sensor, where $r$ is the communication range of the sensor.
\end{definition}
\begin{definition}
{\rm(}LRH{\rm)} Largest regular hexagon inscribed within the communication circle of any sensor.
\end{definition}
\begin{definition}
{\rm(}Beacon point{\rm)}
The position of the anchor that is extracted from the beacon received by a sensor at time $x$ is denoted as a beacon point for the sensor if and only if the sensor does not receive any beacon either in time interval $[x-t_0, x)$ or in time interval $(x, x+t_0]$, where $t_0$ is the waiting time such that $t<t_0<2t$ and $t$ is time interval of periodical broadcasts of beacon by the anchor.
\end{definition}
In this paper we use $Q$ as a sensor as well as a point in the plane that defines the location of that sensor.
The algorithm begins with localizing a sensor $Q$ by random movement of the anchor.
The localized sensor broadcasts its position which is received by the anchor and according to our movement strategy,
the anchor reaches at any point $A$ on the communication circle of the localized sensor $Q$.
The anchor computes the LRH inscribed within the communication circle with $A$ as a vertex as shown in Figure \ref{f:fig7}. At the same time all the other vertices of LRH are also computed by the anchor. Then the anchor starts moving along the LRH and broadcasting beacons with its current position along with all vertices of the LRH at regular interval so that any sensor which receives a beacon, knows the LRH.
If the time interval between two beacons received by a sensor is more than $t_0$, the sensor can easily identify the LRH on which the anchor is moving. At the same time all the neighbors of the localized sensor $Q$ marks at least two beacon points which help them to
compute two probable positions of themselves according to the scheme \cite{Lee2009}. We briefly discuss below how the localization scheme \cite{Lee2009} works.
\begin{figure}[h]
\psfrag{A}{$A$}
\psfrag{C}{$C$}
\psfrag{C'}{\hspace{-1mm}$C'$}
\psfrag{C''}{$C''$}
\psfrag{T}{$T$}
\psfrag{T'}{$T'$}
\psfrag{Q}{\hspace{-1mm}$Q$}
\psfrag{Q'}{\hspace{-1mm}$Q'$}
\psfrag{N}{$N$}
\psfrag{N'}{\hspace{-1mm}$N'$}
    \centering
    \subfloat[$Q$ and $Q'$ are the two possible positions]{\label{f:fig0}\includegraphics[width=0.35\textwidth]{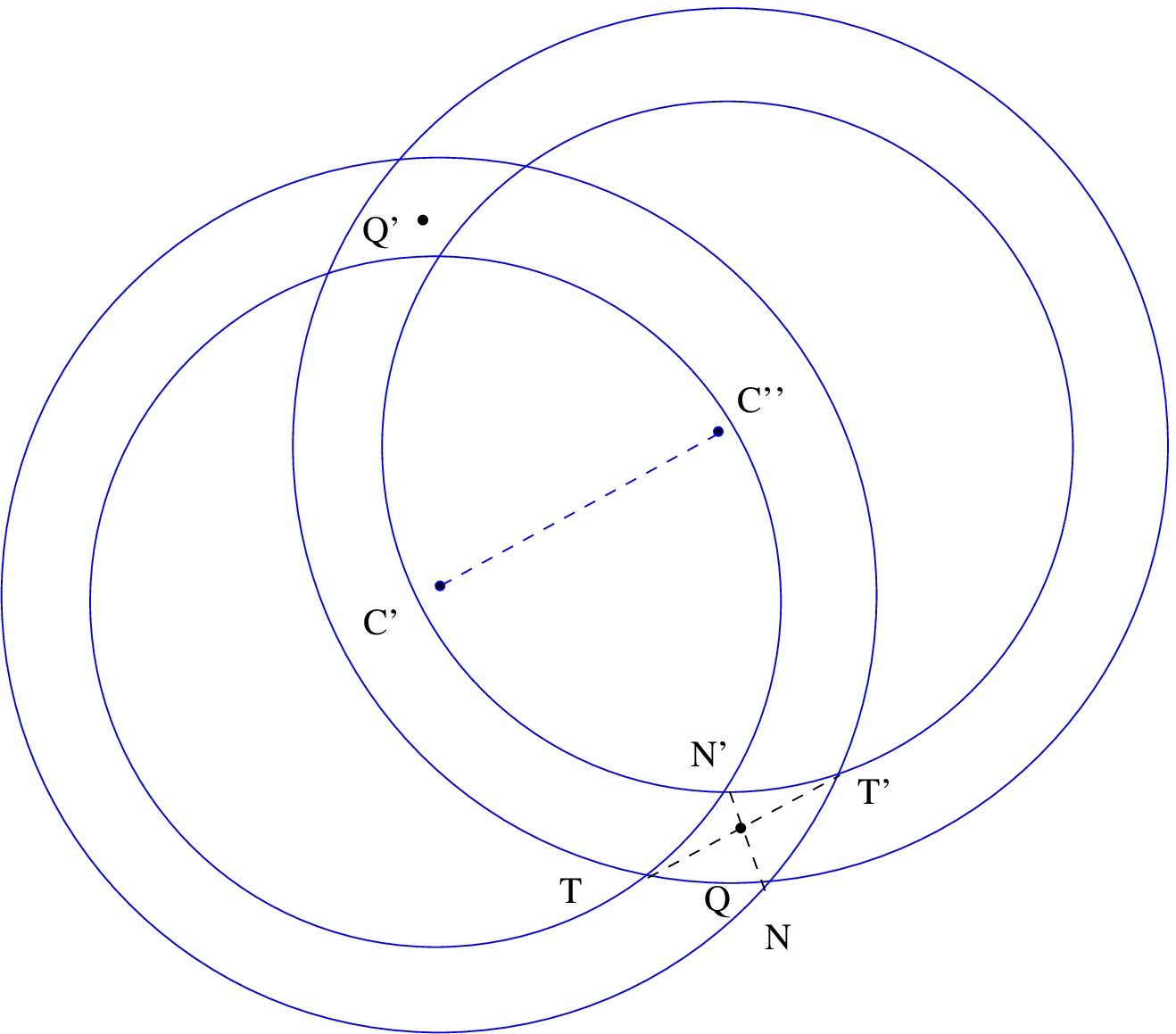}}
    ~~~~~~~
    \subfloat[Sensor chooses its position at $Q$]{\label{f:fig7}\includegraphics[width=0.35\textwidth]{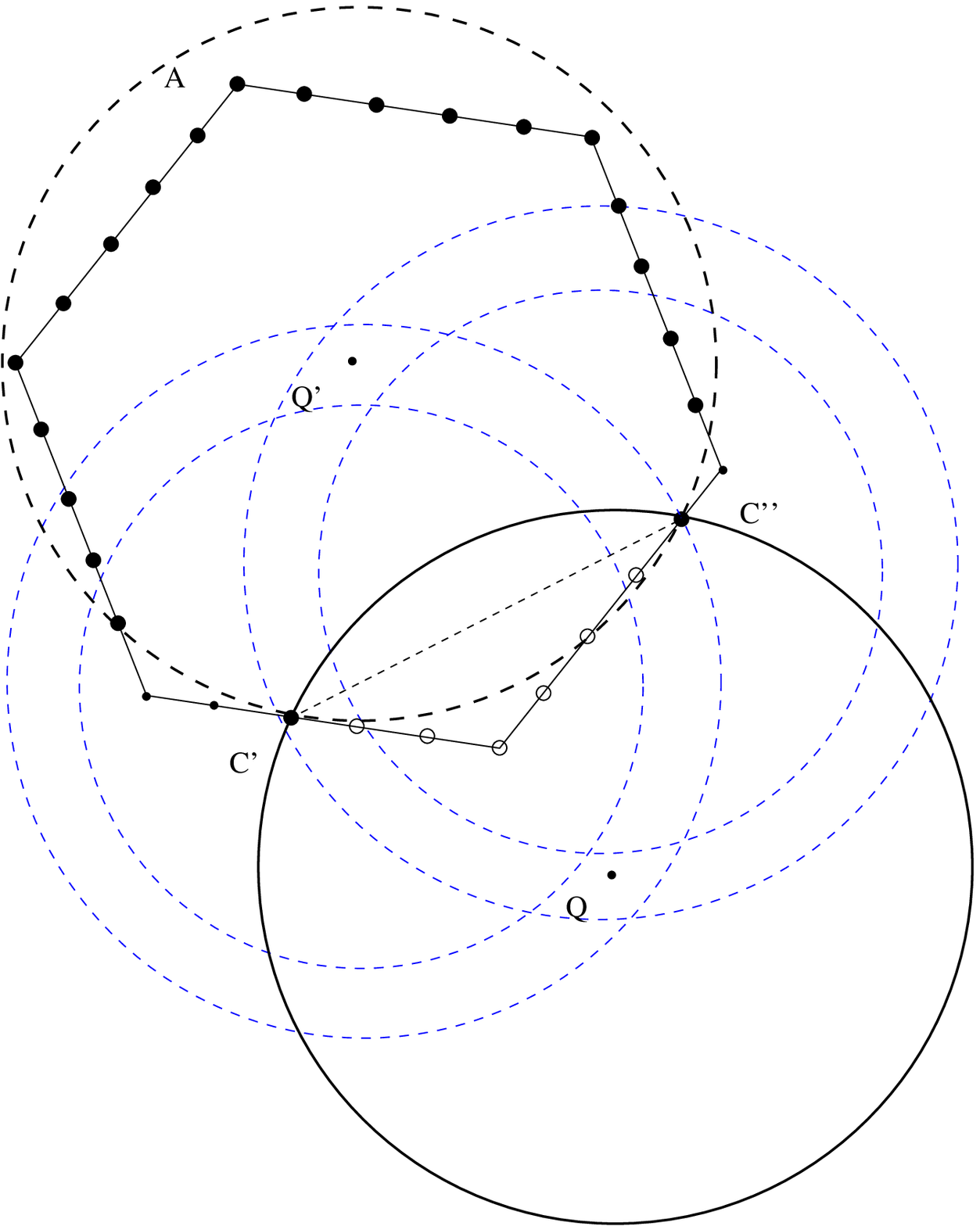}}
\caption{Detection of correct position using hexagonal movement pattern of mobile anchor}
\end{figure}

Let $C'$ and $C''$ be two beacon points marked by a sensor as shown in  Figure \ref{f:fig0}.
Sensor lies on the one of the two intersections of the circular laminae with radius
$(r-u)$ and $r$ centering at the beacon points $C'$ and $C''$,  where u is the beacon distance.
As the intersection is an area, the position of the sensor is considered as the intersection point $Q$ of $NN'$ and $TT'$.
So maximum localization error is equal to $max(NN'/2, TT'/2)$.
Similarly, $Q'$ is another possible position. A third beacon point helps to choose the correct one among $Q$ and $Q'$.
In our work we do not need the third beacon point to localize a sensor. Hexagonal movement strategy helps to
choose the correct one among those two. According to  Figure \ref{f:fig7}, the sensor chooses $Q$ as its
position instead of $Q'$ depending on the received beacons other than those beacon points $C'$ and $C''$. Since sensor
knows the LRH, it finds that if $Q'$ would have been its position than it should have received
beacons at those positions indicated by filled circles on the LRH as shown in  Figure \ref{f:fig7}.
Instead of that, those unfilled circles shown in  Figure \ref{f:fig7} on the LRH are the beacons received by $Q$.
So, sensor selects $Q$ as its position. This selection method will fail only if there exist a common set of beacons
for two different sensor positions, which is not possible.
Following theorems guarantee that all neighbors of any sensor can localize by one complete movement of the
mobile anchor along the LRH around that sensor.

\begin{theorem}
\label{the1:errbd}
Using the scheme \cite{Lee2009} of Lee et al., localization error remains less than $r/2$  for a suitable beacon
distance $u$ if $l\geq (r-u)$, where $l$ is the distance between two beacon points and $r$ is the communication range.
\end{theorem}
\begin{proof}
There are two cases depending on the length of $l$.
\begin{description}
\item[Case 1 ($r-u\leq l \leq 2r-2u$):]
Figure \ref{f:figCase1} illustrates the case.
According to  Figure \ref{f:fig1}, $C'$ and $C''$ be the beacon points received by a sensor such that $r-u\leq l \leq 2r-2u$,
where $C'C''=l$.
\begin{figure}[h]
\psfrag{C'}{$C'$}
\psfrag{C''}{$C''$}
\psfrag{T}{$T$}
\psfrag{T'}{$T'$}
\psfrag{N}{$N$}
\psfrag{N'}{$N'$}
\psfrag{M}{$M$}
    \centering
    \subfloat[Showing the line segment $NN'$]{\label{f:fig1}\includegraphics[width=0.31\textwidth]{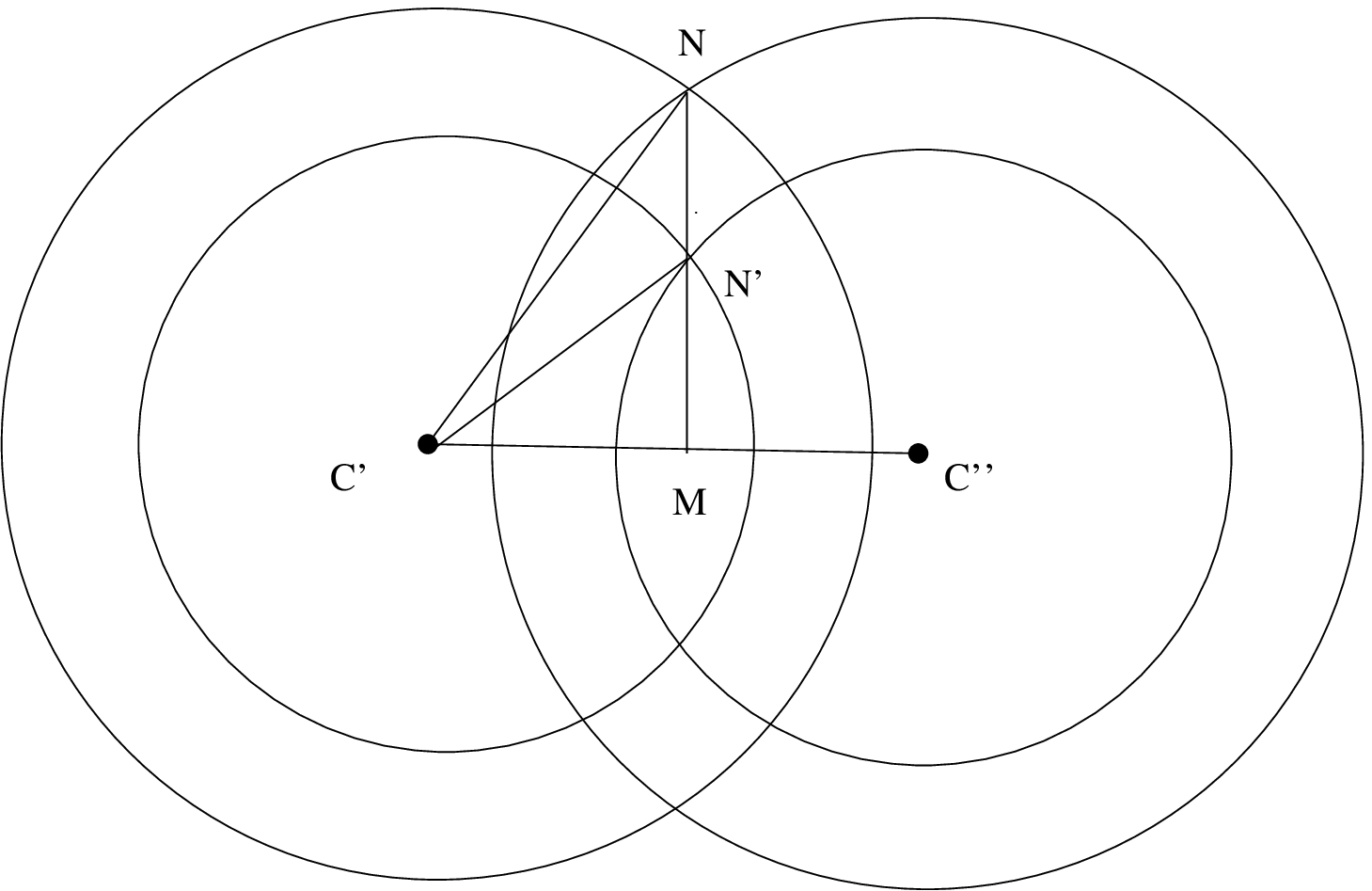}}
    ~~
    \subfloat[$\angle C'TT'=\pi/2$]{\label{f:fig2}\includegraphics[width=0.27\textwidth]{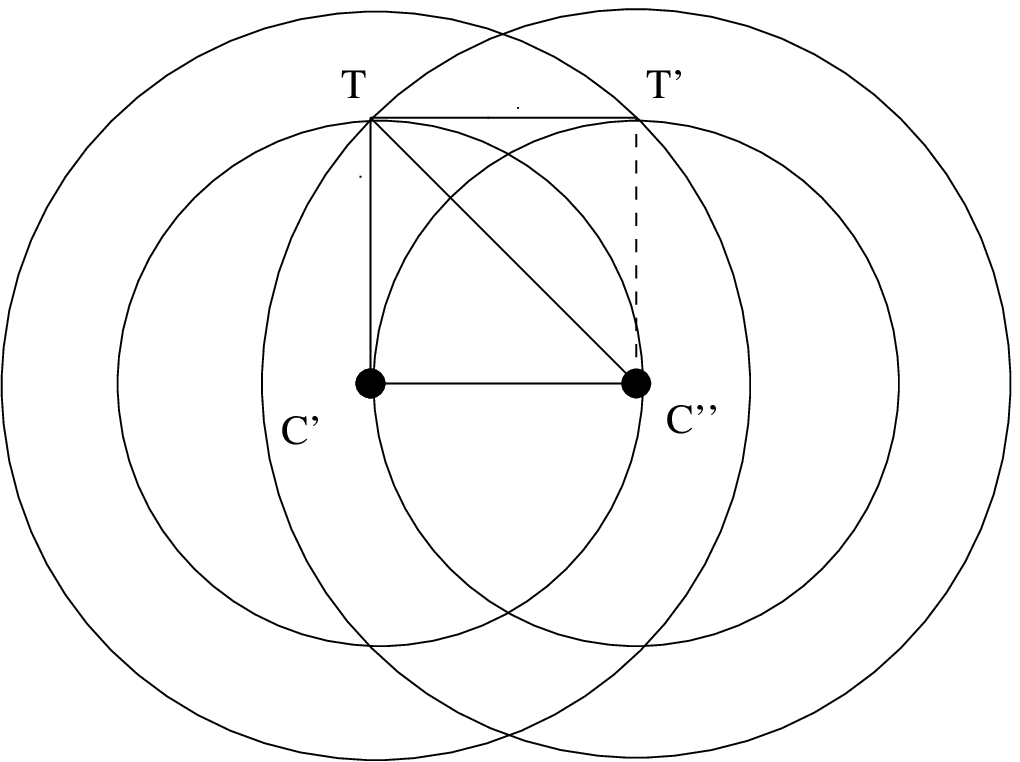}}
    ~~
    \subfloat[$TT'<r$ when $\angle C'TT'>\pi/2$]{\label{f:fig3}\includegraphics[width=0.31\textwidth]{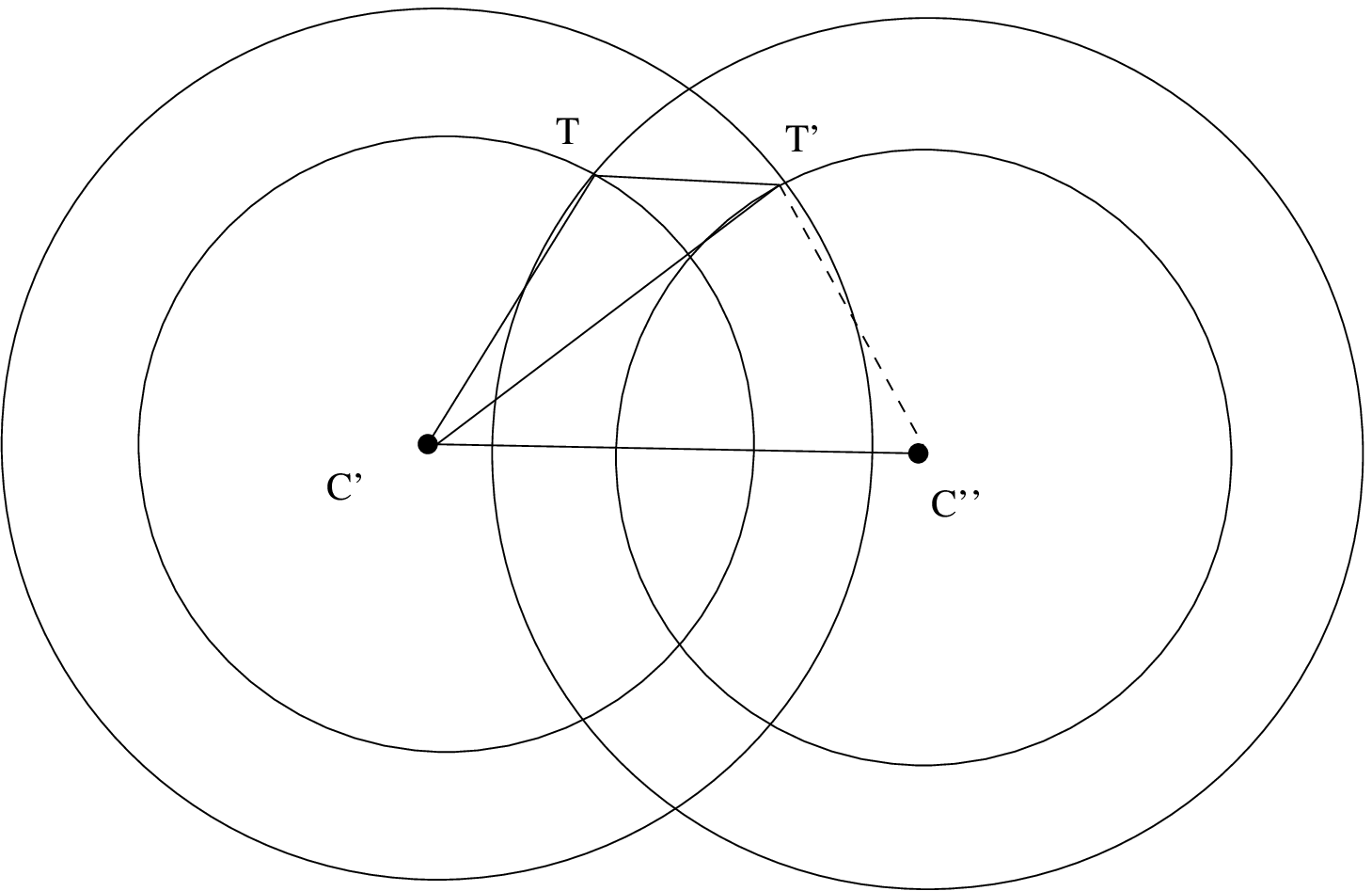}}
\caption{Illustration of case 1}\label{f:figCase1}
\end{figure}

From the figure, $NN'=\sqrt{r^2-(l/2)^2}-\sqrt{(r-u)^2-(l/2)^2} < r$.

Now consider the angle $\angle C'TT'$ in Figure \ref{f:fig2}. Since $TT'$ is parallel
with $C'C''$, so $\angle C'TT'=\pi/2$ implies $\angle TC'C''=\pi/2$. Then from triangle
$\triangle C'C''T$, $(C'T)^2+(C'C'')^2=(C''T)^2$. Since the minimum value of $C'C''$ is $r-u$,
so, $(C'T)^2+(C'C'')^2=(C''T)^2$ implies $\sqrt2(r-u)=r$, which is same as $u=(\sqrt 2-1)r/\sqrt 2$.
Therefore, $\angle C'TT'=\pi/2$ implies $u=(\sqrt 2-1)r/\sqrt 2$.
Hence, $u<(\sqrt 2-1)r/\sqrt 2=r/3.5$ implies $\angle C'TT'>\pi/2$.
\begin{figure}[h]
\psfrag{C'}{$C'$}
\psfrag{C''}{$C''$}
\psfrag{T}{$T$}
\psfrag{T'}{$T'$}
\psfrag{N}{$N$}
\psfrag{N'}{$N'$}
\centering
\includegraphics[width=0.5\textwidth]{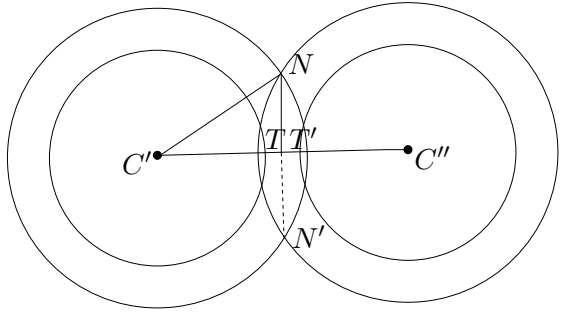}
\caption{Illustration of case 2}\label{f:fig4}
\end{figure}

From triangle $\triangle C'TT'$ in Figure \ref{f:fig3}, we can say, $TT'<r$ when $\angle C'TT'>\pi/2$, i.e.,
$TT'<r$ when $u<r/3.5$. Hence,
\begin{eqnarray}\label{eq:1}
  \max(TT'/2, NN'/2) &<& r/2~ \text{when}~ u<r/3.5
\end{eqnarray}
\item[Case 2 ($l>2r-2u$):] 
From Figure \ref{f:fig4}, $NN'=2\sqrt{r^2-(l/2)^2}$. This is maximum when
$l$ is minimum, i.e., $l=2r-2u$. Now, $NN'<r$ if $2\sqrt{r^2-(r-u)^2}<r$, i.e.,
$NN'<r$ if $u<(2-\sqrt 3)r/2=r/7.5$. Again from Figure \ref{f:fig4}, $TT'<2u$. So, $TT'<r$ if $u<r/2$.
Hence,
\begin{eqnarray}\label{eq:2}
  \max(TT'/2, NN'/2) &<& r/2~ \text{when}~ u<r/7.5
\end{eqnarray}
\end{description}

From the above two inequalities \ref{eq:1} and \ref{eq:2}, for $l\geq (r-u)$, the error bound of the scheme \cite{Lee2009}
is $r/2$ if $u<\min(r/3.5, r/7.5)$, i.e., error is less than $r/2$ if $u<r/7.5$.
\end{proof}

\begin{theorem}
\label{the2:localz}
If an anchor completes its movement along the LRH around a sensor $Q$, then all other sensors lying inside the circle of radius $3r/2$ centering at $Q$ can be localized with error less than $r/2$ for suitable beacon distance $u$,
if $Q$ has been localized within $r/2$ error.
\end{theorem}
\begin{proof}
In Figure \ref{f:fig5}, $Q$ is the calculated position of a sensor and the anchor moves along LRH around $Q$.
Here the LRH is $ABCDEF$.
\begin{figure}[h]
\psfrag{A}{$A$}
\psfrag{B}{$B$}
\psfrag{C}{$C$}
\psfrag{D}{$D$}
\psfrag{C'}{$C'$}
\psfrag{C''}{$C''$}
\psfrag{E}{$E$}
\psfrag{F}{$F$}
\psfrag{G}{$G$}
\psfrag{H}{$H$}
\psfrag{J}{$J$}
\psfrag{P'}{$P'$}
\psfrag{P''}{$P''$}
\psfrag{Q'}{$Q'$}
\psfrag{Q}{$Q$}
    \centering
    \subfloat[Position of the sensor is within the $\triangle QDE$]{\label{f:fig5a}\includegraphics[width=0.4\textwidth]{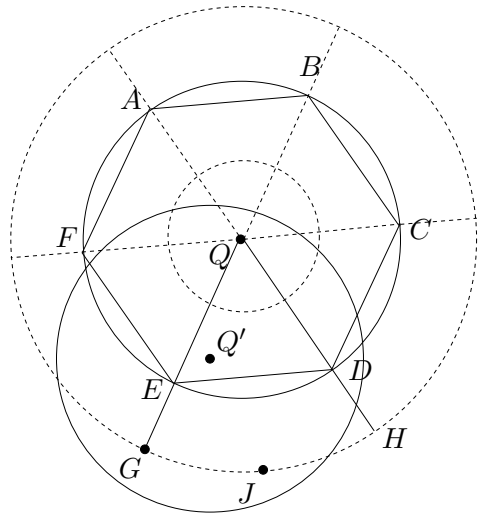}}
     ~~~~~~~~~
    \subfloat[Position of the sensor is within the region bounded by $GE$, $ED$, $DH$ and the arc $HJG$]{\label{f:fig5b}\includegraphics[width=0.35\textwidth]{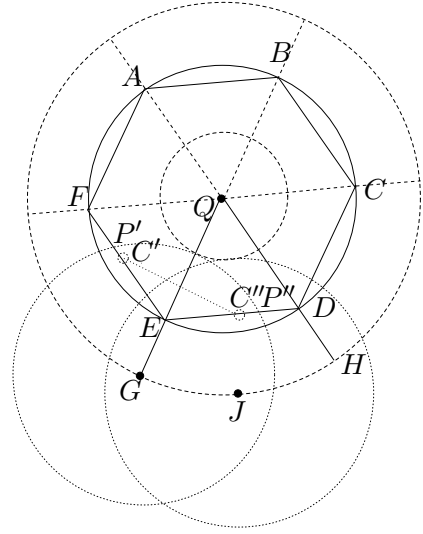}}
\caption{Position of a sensor within the sector $QGH$}\label{f:fig5}
\end{figure}
Due to localization error, $Q$ may lie anywhere within the smallest dotted circle in Figure \ref{f:fig5}
with radius $r/2$ centering at $Q$.
We prove that if a sensor lies anywhere within the largest dotted circle with radius $3r/2$ centering at $Q$
then it marks at least  two beacon points such that distance between them is at least $(r-u)$, where $u$ is the beacon
distance. Length of each side of the regular hexagon $ABCDEF$ is $r$ and $QA=QB=QC=QD=QE=QF=r$. The anchor starts
its movement from a vertex of the hexagon and broadcasts beacon maintaining beacon distance $u=r/k$ for some integer $k$,
which ensures that the anchor broadcasts beacon at every vertex of the LRH along with some other points on the LRH.

We divide the largest dotted circle shown in Figure \ref{f:fig5} in six symmetric sectors by
joining each vertex of the hexagon with $Q$ and then extending the line up to the boundary of the largest circle.
So any sensor lies within the largest circle must lie in any one of the six symmetric sectors. If a sensor lies
on common boundary on any two sectors then it can be considered within any of the sector.
Without loss of generality, let a sensor lie in the sector bounded by the line segments $GQ$, $QH$ and
by the arc $HJG$, where $J$ is the mid-point of $GH$ on the arc. We further divide the sector in two parts i.e., $\triangle QED$ and the region $GEDHJG$ bounded by $GE$, $ED$, $DH$ and the arc $HJG$, where $\triangle QED$ defines the triangle with vertices $Q, E$ and $D$. There are two cases based on the position of a sensor belonging to $\triangle QED$ or region $GEDHJG$.
\begin{description}
  \item[Case 1.] 
  $Q'$ is the location of sensor inside $\triangle QED$ as shown in Figure \ref{f:fig5a}, where $\triangle QED$ is an equilateral triangle with side length $r$. Then irrespective of the position of $Q'$, $Q'E<r$, $Q'D<r$. Hence the communication circle of $Q'$ intersects $CD$ and $EF$/$FA$ or intersects $EF$ and $CD$/$CB$. So distance $l$ between two beacon points is at least $r$ i.e., $l>(r-u)$. So $Q'$ can be localized within $r/2$ error using Theorem \ref{the1:errbd}.
  \item[Case 2.] $Q'$ is the location of sensor inside the region $GEDHJG$. The distance $l$ between the beacon points decreases, if position of $Q'$ moves away from the center $Q$ of the LRH. As any point on the arc $GJH$ is furthest from $Q$, hence, it is sufficient to show $l>(r-u)$ when $Q'$ lies on the arc $GJH$. Due to the symmetric nature of the arcs $JG$ and $JH$, it is sufficient to show for positions of $Q'$ on arc $JG$. Let position of $Q'$ is at $J$. Since $JE=JD<r$, circle with radius $r$ centering at $J$ intersects $CD$ and $EF$ as shown in Figure \ref{f:fig5b}. Hence $Q'$ marks two beacon points such that $l\geq r >(r-u)$.
      One can easily verify that as position of $Q'$ changes along the arc $JG$ towards $G$, distance $l$ between two beacon points decreases and $l$ is least when $Q'$ is at $G$. So, it is sufficient to show that $l>(r-u)$ when $Q'$ is at $G$.
      As shown in Figure \ref{f:fig5b}, communication circle of $Q'$, i.e., the circle centering at $G$, intersects $DE$ and $EF$ at $P'$ and $P''$ respectively i.e., $GP'=GP''=r$. Let $C'$, $C''$ be the beacon points considering the worst case such that $P''C''=P'C'= u$.
      We can write $EC'=EP'-u$ and $EC''=EP''-u$. Let $EP'=EP''=L$. Then from triangle $\triangle GEP''$ of Figure \ref{f:fig5b},
      $L^2+(r/2)^2-2L(r/2)\cos (2\pi/3)=r^2$. Solving this equation, we get $L=(\sqrt 13-1)r/4$. Now from  $\triangle EC'C''$,
      $C'C''=\sqrt{(L-u)^2+(L-u)^2-2(L-u)^2\cos (2\pi/3)}$, i.e., $C'C''=\sqrt 3\left((\sqrt 13-1)(r/4)-u\right)$. So, $C'C'' \geq (r-u)$ if $\sqrt 3((\sqrt 13-1)r/4-u)\geq r-u$, i.e., $u\leq (\sqrt 3(\sqrt 13-1)-4)r/(4(\sqrt 3-1))$, i.e., $u\leq r/5$.
\end{description}
So, we can conclude that for $u\leq r/5$, all the sensors which lie within the circle of radius $3r/2$ and centering at $Q$, can be localized with error less than $r/2$.
\end{proof}

\begin{corollary}
\label{cor:cor}
If $u\leq r/7.5$, all the sensors which lie within the circle of radius $3r/2$ centered at $Q$,
can be localized with error less than $r/2$ after mobile anchor completes its movement along the LRH around the sensor $Q$.
\end{corollary}
\begin{proof}
Follows from Theorem \ref{the1:errbd} and Theorem \ref{the2:localz}.
\end{proof}

\begin{theorem}
\label{thm3:nbrslclztn}
If a mobile anchor completes its movement along the LRH around a sensor which is localized within $r/2$ error, then all its neighbors
are localized.
\end{theorem}
\begin{proof}
If  a sensor is localized within $r/2$ error, its neighbors lie within the circle of radius $3r/2$ centering the sensor.
Hence the statement follows by Theorem \ref{the2:localz}.
\end{proof}

\section{Distributed Algorithm for Path Planning}
\label{sec:algomodel}
We assume sensors form a connected network. The number of sensors in the network is not an input of our algorithm.
Each sensor has unique id and knows the id of its one hop neighbors. Set of neighbors of a sensor $i$ is denoted by $nbd(i)$.
We define {\it NLN-degree}$(i)$ as the number of non-localized neighbors of a sensor $i$. Initially NLN-degree$(i)=|nbd(i)|$, where
$|nbd(i)|$ is the cardinality of the set $nbd(i)$. We assume at the beginning of localization, a sensor $i$ is localized itself by
random movement of the mobile anchor such that the localization error is within $r/2$, where $r$ is the communication range. This can be done by choosing beacon points such that they are at least $r-u$ distance apart, where $u$ is the beacon distance. After localization, $i$
broadcasts its position.
Whenever the mobile anchor hears the position of $i$, it starts moving along the sides of the LRH of $i$. Computation of LRH is explained in section \ref{sec:tech1}. Hence each $j\in nbd(i)$ localizes within the error bound $r/2$ by Theorem \ref{thm3:nbrslclztn} and broadcasts the position to $nbd(j)$. After receiving position of a neighbor $j\in nbd(k)$, sensor $k$ updates NLN-degree$(k)$ by NLN-degree$(k)-1$. During the process of completion of LRH, each sensor which receives a neighbor's position, keeps updating its NLN-degree. So, when all neighbors of sensor $i$ are localized at the end of a LRH movement, each $j\in nbd(i)$ has computed their NLN-degree$(j)$. After completing LRH, the anchor moves $r/2$ distance towards $i$. Moving $r/2$ distance towards $i$ is required to ensure communication between the anchor and $i$ because the positional error is bounded by $r/2$ according to Theorem \ref{the2:localz}. Then anchor sends a message to $i$ for its next destination of movement.
To do so, $i$ sends a request to all $j\in nbd(i)$ for NLN-degree$(j)$ along with their positions. Then $i$ selects
a neighbor sensor $j'$ that achieves the value $\max\{$NLN-degree$(j) | j\in nbd(i)\}$ and sends position of $j'$ with NLN-degree$(j')$ to the anchor for next destination of movement. The anchor moves to the closest point of the communication circle of $j'$ and starts the LRH movement around $j'$ to continue localization. The mobile anchor maintains a STACK along with its operations PUSH, POP and variable TOP with usual dynamic stack data structure during its travel through the connected network.
Initially the STACK is empty. Before making LRH movement around a sensor $i$, the anchor PUSH id $i$ on the STACK. If $\max\{$NLN-degree$(j) | j\in nbd(i)\}=0$ then anchor POP $i$ from the STACK. When the STACK becomes empty, the algorithm terminates, otherwise anchor revisits the communication circle of the sensor $i'$ (say), which is at the TOP of the STACK and then sends a message to $i'$ for its next destination of movement.
Mobile anchor decides its path according to the Algorithm \ref{alg:ALG}: \textsc{HexagonalLocalization}.
\begin{algorithm}[]
\caption{\textsc{HexagonalLocalization}}
\begin{algorithmic}[1]
                    \STATE Mobile anchor localizes a sensor by its random movement then PUSH id of the sensor into the STACK.
                    \STATE \label{lin2} Computes LRH centering at the sensor whose id $i$ is at the TOP of the STACK and broadcasts beacons periodically with period $t$ until the LRH movement completes.
                    \STATE \label{lin3} The anchor moves $r/2$ distance towards $i$ and sends a message to $i$ for next destination of its movement.
                    \STATE Sensor $i$ sends a message to all $j\in nbd(i)$ for NLN-degree$(j)$ along with their positions.
                    \STATE On receiving the replies, sensor $i$ selects
                        a neighbor sensor $j'$ that achieves the value $\max\{$NLN-degree$(j) | j\in nbd(i)\}$ and sends position of $j'$ with NLN-degree$(j')$ to the anchor for next destination of movement.
                    \STATE If NLN-degree$(j')>0$ then the anchor PUSH $j'$ into the STACK and moves to the closest point of the communication circle of $j'$ and executes step \ref{lin2}, otherwise POP from the STACK.
                    \STATE The algorithm terminates if STACK is empty, otherwise the anchor revisits the sensor whose id is at the TOP of the STACK and executes step \ref{lin3}.
\end{algorithmic}
\label{alg:ALG}
\end{algorithm}

\subsection{Correctness and complexity analysis}
\begin{theorem}
Algorithm \ref{alg:ALG} ensures localization of all sensors in a connected network.\label{th:correctness}
\end{theorem}
\begin{proof}
We prove correctness of the algorithm by  method of contradiction. Let us assume that one sensor $i$, (say), is not localized but the algorithm terminates, i.e., the STACK becomes empty. Since all sensors are in a connected network, then there exist at least one sensor $j$, (say), which is a neighbor of $i$ and is localized. According to the algorithm, $j$ localizes itself by marking beacon points on a LRH movement of the anchor around one of its neighbors, say $k$. At this moment $k$ is the TOP of the STACK. As $i$ is not localized, $j$ would send non zero NLN-degree$(j)$ and its position to $k$. Since other sensors are localized, NLN-degree$(j)$ is the maximum among those received by $k$. Hence $k$ should send the position of $j$ to the anchor for the next destination of movement according to algorithm \ref{alg:ALG}. Whenever anchor visits $j$ and makes the LRH movement, $i$ becomes localized. So $k$ cannot be popped without pushing $j$ which implies localization of $i$. Hence $k$ is in the STACK until $i$ is localized. It contradicts our assumption that the STACK is empty but $i$ is not localized. Hence proved.
\end{proof}
\begin{theorem}
The time complexity of Algorithm \ref{alg:ALG} is $O(|V|+|E|)$.\label{th:timecomplexity}
\end{theorem}
\begin{proof}
To analyze complexity of Algorithm \ref{alg:ALG}, we calculate maximum travel distance of the mobile anchor to localize all sensors in any connected graph $G=(V,E)$ topology. Here $V$ and $E$ are the set of vertices corresponding to the sensors and the set of edges of $G$ respectively. If degree of the graph decreases then the anchor needs to make more LRH movements. The line graph achieves lowest degree among connected graphs and hence the anchor attends maximum LRH movements to localize all sensors. In the worst case, for a line graph our algorithm matches with DFS visit of $G$. In this case anchor has to make LRH movement around $|V|-1$ sensors if it initiates its movement from one end of the graph. Total LRH movement is $6r(|V|-1)$ since each LRH movement equals to perimeter $6r$ of LRH, where $r$ is the communication range of the sensors. In addition to this, anchor has to move maximum $2r(|E|-1)$ distance to reach other sensors and then to return to the initiator.  Hence both time complexity and complexity in terms of distance traveled by the anchor are same and equal to $O(|V|+|E|)$.
\end{proof}
\section{Path Planning for Rectangular Region}
\label{sec:tech2}
In this section we describe path planning of a mobile anchor, which covers a given rectangular region to ensure localization of all sensors deployed
on the region. Let $ABCDEF$ be a regular hexagon with side length $r$ and center $O$ as shown in  Figure \ref{f:fig8}.
If communication circle of a sensor touches the hexagon then we say that the sensor is within the coverage area of the hexagon.
Obviously the area inside the hexagon is a part of the coverage area.  As shown in Figure \ref{f:fig8}, $A'B'C'D'E'F'$ is a regular hexagon with side $2r$ and center $O$. The communication circles of the sensors located at the vertices of the larger hexagon touches the smaller hexagon at one vertex. Communication circle of any sensor which lies within or on the larger hexagon except at the vertices, intersects the smaller hexagon at at least two points.
\begin{figure}[h]
\psfrag{O}{$O$}
\psfrag{A}{$A$}
\psfrag{M}{\hspace{-1mm}$M$}
\psfrag{M'}{\hspace{-1mm}$M'$}
\psfrag{B}{$B$}
\psfrag{C}{$C$}
\psfrag{D}{$D$}
\psfrag{E}{$E$}
\psfrag{F}{$F$}
\psfrag{A'}{$A'$}
\psfrag{B'}{$B'$}
\psfrag{C'}{$C'$}
\psfrag{D'}{$D'$}
\psfrag{E'}{$E'$}
\psfrag{F'}{$F'$}
\psfrag{A''}{$A''$}
\psfrag{B''}{\hspace{-1mm}$B''$}
\psfrag{C''}{$C''$}
\psfrag{D''}{$D''$}
\psfrag{E''}{$E''$}
\psfrag{F''}{$F''$}
  \centering
    \subfloat[The hexagon $A'B'C'D'E'F'$ is the coverage area when mobile anchor moves along the hexagon  $ABCDEF$]{\label{f:fig8}\includegraphics[width=0.35\textwidth]{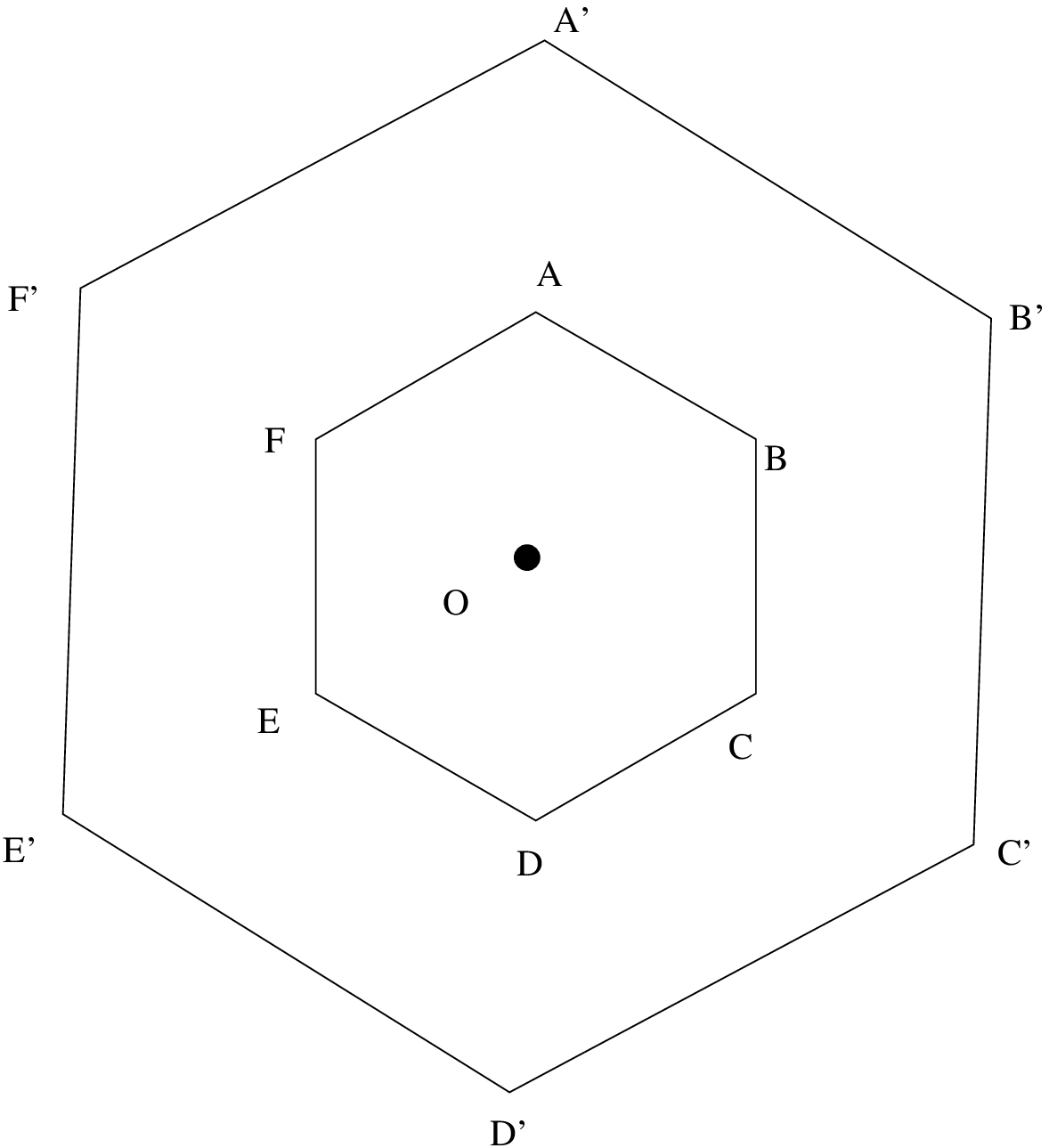}}
     ~~~~~~~~~
    \subfloat[The coverage area $A''B''C''D''E''F''$ ensures two beacon points on $ABCDEF$]{\label{f:fig9}\includegraphics[width=0.4\textwidth]{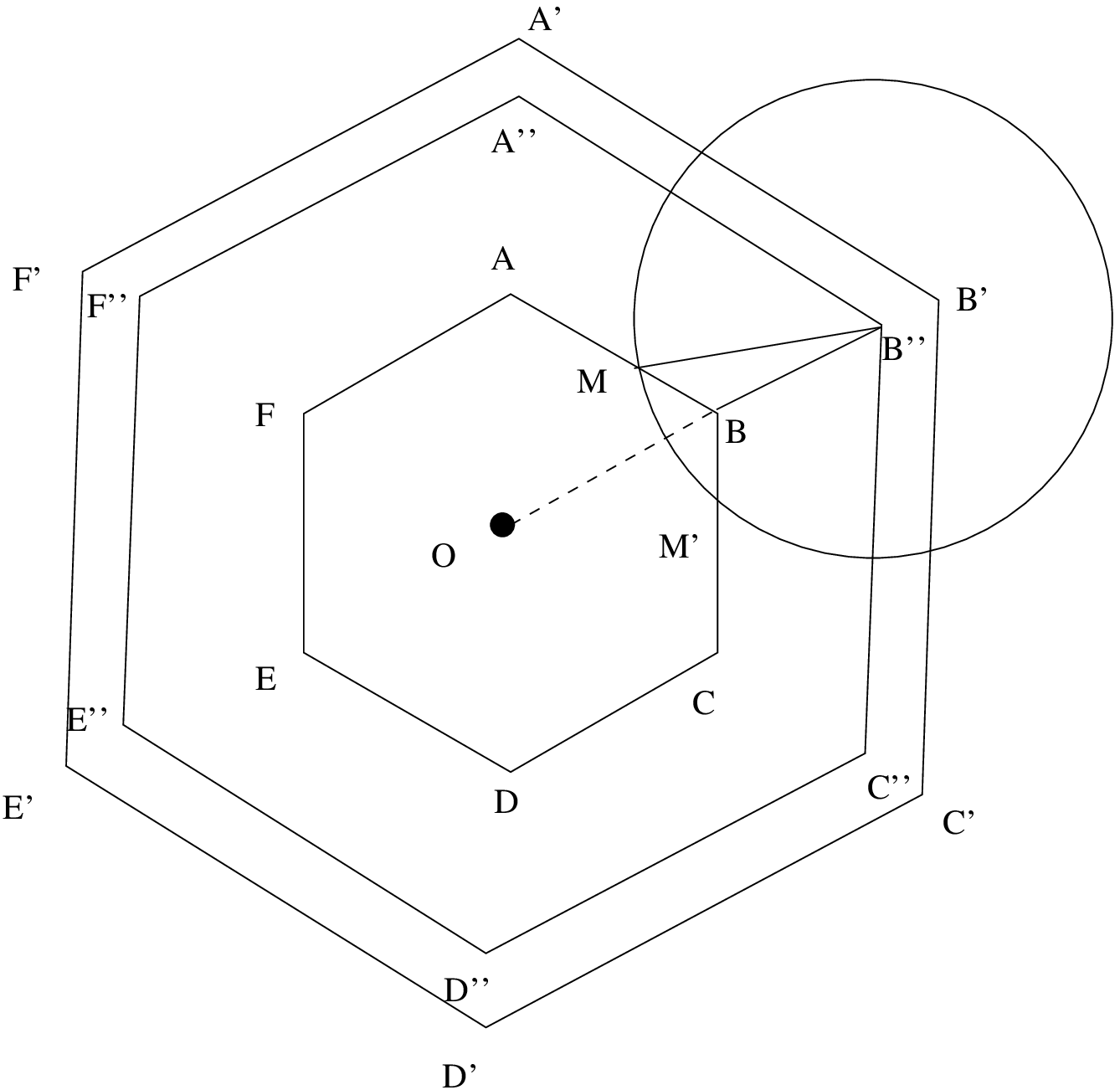}}
\caption{Showing coverage area when a mobile anchor moves along the regular hexagon $ABCDEF$}
\end{figure}
So $A'B'C'D'E'F'$ forms a coverage area of $ABCDEF$.
Let the anchor moves along $ABCDEF$ and broadcasts beacons periodically with time period $t$. For localization with the scheme \cite{Lee2009}, each sensor should mark at least two beacon points on $ABCDEF$. To ensure marking of two beacon points, we reduce the coverage area by reducing the hexagon $A'B'C'D'E'F'$ to $A''B''C''D''E''F''$ with side $2r-X$ as shown in  Figure \ref{f:fig9}. Now, we have to find a suitable value of $X$ such that all sensors located inside $A''B''C''D''E''F''$ can be localized.
As the vertices of $A''B''C''D''E''F''$ are the furthest points from $ABCDEF$, so if we find $X$ in such a way that a sensor, which lies at any vertex of $A''B''C''D''E''F''$, marks at least two beacon points, then so does all other sensors, which lies on or inside $A''B''C''D''E''F''$ for that same $X$.

Let $BM'=u$ as shown in Figure \ref{f:fig9}. Then the sensor at $B''$ marks two beacon points at $B$, $M'$, if $B''M'=r$.
Let $BB''=r-X$. From the triangle  $\triangle BMB''$, $(r-X)^2+u^2-2(r-X)u\cos (2\pi/3)=r^2$, which implies $X=r+u/2-\frac{\sqrt{4r^2-3u^2}}{2}$.
Approximately $X=u/2$ as $u<<r$. Hence, if mobile anchor moves along a regular hexagon of side $r$, then all the sensors which lie within a regular
hexagon with side $2r-X$, mark at least two beacon points, where $X\geq r+u/2-\frac{\sqrt{4r^2-3u^2}}{2}$.

As shown in Figure \ref{f:fig10}, the mobile anchor moves along the blue solid lines starting from $A$ to cover the rectangle $R_1R_2R_3R_4$.
The movement terminates at $Z$.  The directions of movement are shown by red dashed arrows. We now compute the path length of the movement. According to Figure \ref{f:fig10}, each hexagon with blue solid lines and side $r$ covers one larger hexagon with black dotted lines whose width is equal to \\ $H_2H_4=\sqrt{(2r-X)^2+(2r-X)^2-2(2r-X)^2\cos (2\pi/3)}=\sqrt3(2r-X)$.
\begin{figure}[h]
\begin{center}
\psfrag{Z}{$Z$}
\psfrag{R}{$R$}
\psfrag{A}{$A$}
\psfrag{G}{$G$}
\psfrag{H}{$H$}
\psfrag{1}{\hspace{1mm}$_1$}
\psfrag{2}{\hspace{1mm}$_2$}
\psfrag{3}{\hspace{1mm}$_3$}
\psfrag{4}{\hspace{1mm}$_4$}
\includegraphics[width=0.80\textwidth]{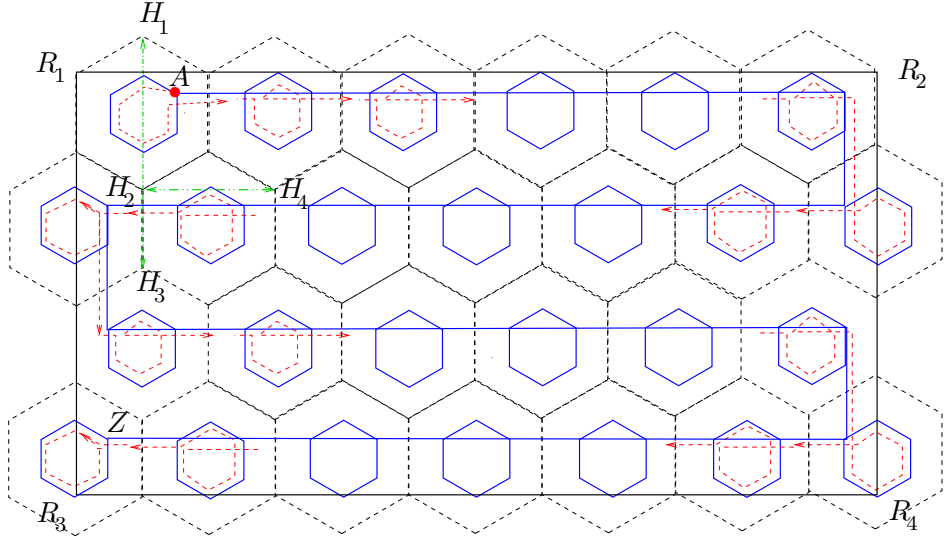}
\caption{\label{f:fig10} Showing path planning to cover a rectangular region $R_1R_2R_3R_4$}
\end{center}
\end{figure}
If $L$ is the width of the rectangle, then the number of hexagon in a row would be $\left\lceil{\frac{L}{\sqrt3(2r-X)}}\right\rceil$.
Two rows of hexagon cover an area with height equal to $H_1H_3=H_1H_2+H_2H_3=3(2r-X)$. Hence, if $L$ is the hight of
the rectangle, then the number of rows of hexagon would be $\left\lceil{\frac{L}{3(2r-X)/2}}\right\rceil$. Then total number of
hexagon with side $r$ required to cover a $L\times L$ rectangle is $\left\lceil{\frac{L}{\sqrt3(2r-X)}}\right\rceil\left\lceil{\frac{2L}{3(2r-X)}}\right\rceil$.
Total path length required to reach all the hexagons in a row by moving from one hexagon to another hexagon is at most $L$.
So, total distance traversed to connect the hexagons for all rows is equal to $L\times$ no. of rows, i.e., $L\left\lceil{\frac{2L}{3(2r-X)}}\right\rceil$.
Total path length required to reach all the rows by moving from one row to another row is at most $L$ again.
In each pair of consecutive rows, an extra hexagon is needed i.e., an extra $3r$ movement is required in each row on an average.
So, total path traversed to cover $L\times L$ rectangle is equal to

$D_{Hexagon}=\left\lceil{\frac{L}{\sqrt3(2r-X)}}\right\rceil\left\lceil{\frac{2L}{3(2r-X)}}\right\rceil6r+L\left\lceil{\frac{2L}{3(2r-X)}}\right\rceil+
3r\left\lceil{\frac{2L}{3(2r-X)}}\right\rceil+L$.

\subsection{Comparison with existing schemes}

Theoretical comparison with existing schemes \cite{Huang2007,Koutsonikolas2007,Chia-Ho-Ou2013} in terms of path length of mobile anchor
is presented in this section. Let $L$ be the length of the sides of a square region and $r$ be the communication range of sensors and anchor.
To normalize all the expressions of path length given in \cite{Chia-Ho-Ou2013}, we have taken $X=r/k$ for some integer $k>1$.
We have denoted path length of mobile anchor in \cite{Chia-Ho-Ou2013} by $D_{Chia-Ho-Ou}$. Path lengths of mobile anchor in different schemes in \cite{Koutsonikolas2007} are denoted by $D_{Scan}$, $D_{Doublescan}$ and $D_{Hilbert}$. Path lengths of mobile anchor in different schemes in \cite{Huang2007} are denoted by $D_{Circles}$ and $D_{S-curves}$. We have denoted the path length of mobile anchor of our proposed scheme by $D_{Hexagon}$. So, path lengths of different schemes according to \cite{Chia-Ho-Ou2013} are given by,

$D_{Chia-Ho-Ou}=(L+2r)\left(\left\lceil\left(\frac{L+2r}{r-r/k}\right)\right\rceil+1\right)+(r-r/k)\left\lceil\left(\frac{L+2r}{r-r/k}\right)\right\rceil$. This is same as $D_{Scan}$.

$D_{Doublescan}=2\left[\left(\frac{L+r+r/k}{2(r-r/k)}+1\right)(L+2r)+L+r+r/k\right]$.

$D_{Hilbert}=\left(\frac{L+2r}{r-r/k}\right)^2(r-r/k))$.

$D_{Circles}=N^2 \pi(r-r/k)$+L, where $N=\frac{L/sqrt(2)-r}{(r-r/k)}$.

Expression of $D_{Circles}$ is modified in such a way that it can also cover the corner points of the rectangle, instead of the largest circle inscribed within the rectangle.

$D_{S-curves}=\left[\left(\frac{L+r+r/k}{3r/2}+1\right)\frac{L+2r}{2}\pi\right]+(L+r+r/k)+(r-r/k)\pi/2$.\\
The path length for our proposed scheme is

$D_{Hexagon}=\left\lceil{\frac{L}{\sqrt3(2r-r/k)}}\right\rceil\left\lceil{\frac{2L}{3(2r-r/k)}}\right\rceil6r+L\left\lceil{\frac{2L}{3(2r-r/k)}}\right\rceil+3r\left\lceil{\frac{2L}{3(2r-r/k)}}\right\rceil+L$.

In all the above expressions, the highest order term is $L^2/r$. So, we have compared the coefficients of $L^2/r$ as shown in Table \ref{table:1} to decide minimality of path length theoretically.

\begin{table}[h]
\center
\caption{Coefficients of higher order term $L^2/r$ in the expressions of different schemes}
\label{table:1}
\begin{tabular}{c|c|c|c|c|c}
\hline
\hline
$D_{Hexagon}$ & $D_{Chia-Ho-Ou}$        & $D_{Doublescan}$ & $D_{Hilbert}$      & $D_{S-curves}$     &$D_{Circles}$   \\
&&&&&\\
\hline
$\frac{1}{\sqrt 3}\left(\frac{2k}{2k-1}\right)\left(\frac{2k}{2k-1}+\frac{1}{\sqrt 3}\right)$ & $\frac{k}{k-1}$ & $\frac{k}{k-1}$  & $\frac{k}{k-1}$    & $\frac{\pi}{3}\left(\frac{k}{k-1}\right)$ & $\frac{\pi}{2}\left(\frac{k}{k-1}\right)$\\
\hline
\hline
\end{tabular}
\end{table}

In $D_{Chia-Ho-Ou}$, $D_{Doublescan}$, and $D_{Hilbert}$, the coefficient of higher order term $L^2/r$ is equal to $\frac{k}{k-1}$.
For $D_{S-curves}$, the coefficient of $L^2/r$ is $\frac{\pi}{3}\left(\frac{k}{k-1}\right)$. The coefficient of $L^2/r$ in $D_{Circles}$ is $\frac{\pi}{2}\left(\frac{k}{k-1}\right)$.
The coefficient of $L^2/r$ in $D_{Hexagon}$ is $\frac{1}{\sqrt 3}\left(\frac{2k}{2k-1}\right)\left(\frac{2k}{2k-1}+\frac{1}{\sqrt 3}\right)$.
We show that coefficient of $L^2/r$ in $D_{Hexagon}$ is lesser than all the above mentioned coefficients for all $k$.
Among the existing path planning schemes $Circles$ does best in terms of path length to cover a circular region. But for a square region, larger circles are needed to cover the corners of the square, which increases path length and makes $Circkes$ a bad strategy for covering a square.
Since $\frac{\pi}{2}\left(\frac{k}{k-1}\right)$>$\frac{\pi}{3}\left(\frac{k}{k-1}\right)$>$\frac{k}{k-1}$ for all $k>1$, so if we can show
$\frac{k}{k-1}$>$\frac{1}{\sqrt 3}\left(\frac{2k}{2k-1}\right)\left(\frac{2k}{2k-1}+\frac{1}{\sqrt 3}\right)$ for all $k$, then we can claim the minimality of path length of our scheme as the coefficient of the highest order term $L^2/r$ of our scheme is least among all the existing ones.
One can easily verify that $\frac{k}{k-1}$>$\frac{1}{\sqrt 3}\left(\frac{2k}{2k-1}\right)\left(\frac{2k}{2k-1}+\frac{1}{\sqrt 3}\right)$ for $2\leq k \leq9$. Again, for all $k \geq 10$, $\frac{1}{\sqrt 3}\left(\frac{2k}{2k-1}\right)\left(\frac{2k}{2k-1}+\frac{1}{\sqrt 3}\right)<1$ but $\frac{k}{k-1}>1$. Hence for any $r$, there exists $L$ such that path length of $D_{Hexagon}$ is minimum.

\section{Simulation Results}
\label{sec:sim}
We have used MATLAB platform to study the performances of our proposed schemes. We have randomly generated a connected graph of
sensors in a 50 meter $\times$ 50 meter square region. According to the section \ref{sec:tech1}, we have taken values of beacon
distance $u < r/7.5$. Following Figure \ref{f:fig11} and Figure \ref{f:fig12} show the hexagonal movement path (red solid lines)
of a mobile anchor through the network during localization, where black bold lines are transition path between LRHs.
The blue crosses and red circles are the actual and calculated positions of the sensors respectively and pairs are joined by dotted lines.
The number of sensors $n$, communication ranges $r$, path length $D$ and average positioning error are given in the caption of each figure.
All the values showing in the tables and figures are in meter. As the area remains fixed and connectivity is maintained by increasing
communication range of sensors and the anchor, path length does not decrease much with the number of sensors.
\begin{figure}[h]
\centering
\subfloat[$n=25$, $r=20$, $D=966$, Average error=1.42]{\label{f:fig11a} \includegraphics[width=0.5\textwidth]{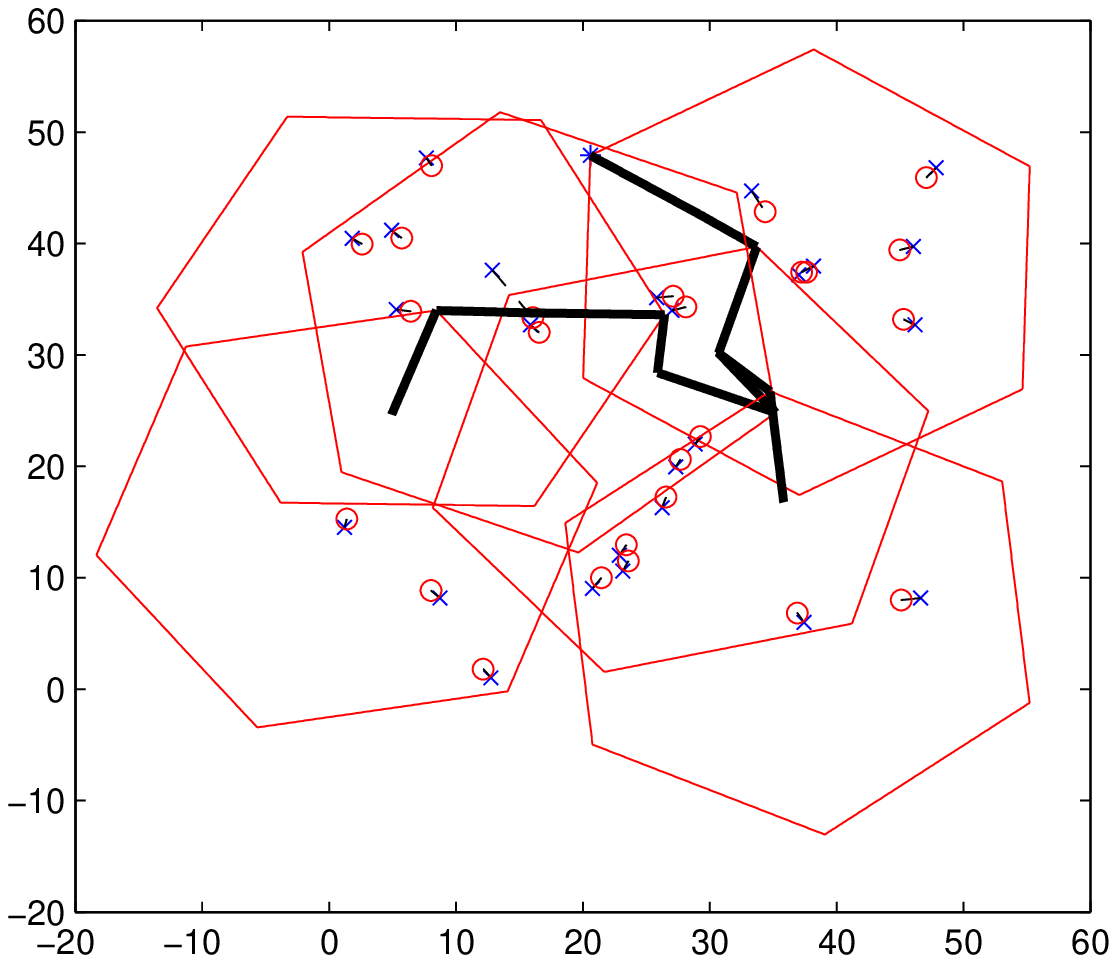}}
\subfloat[$n=50$, $r=15$, $D=1097$, Average error=1.09]{\label{f:fig11b} \includegraphics[width=0.55\textwidth]{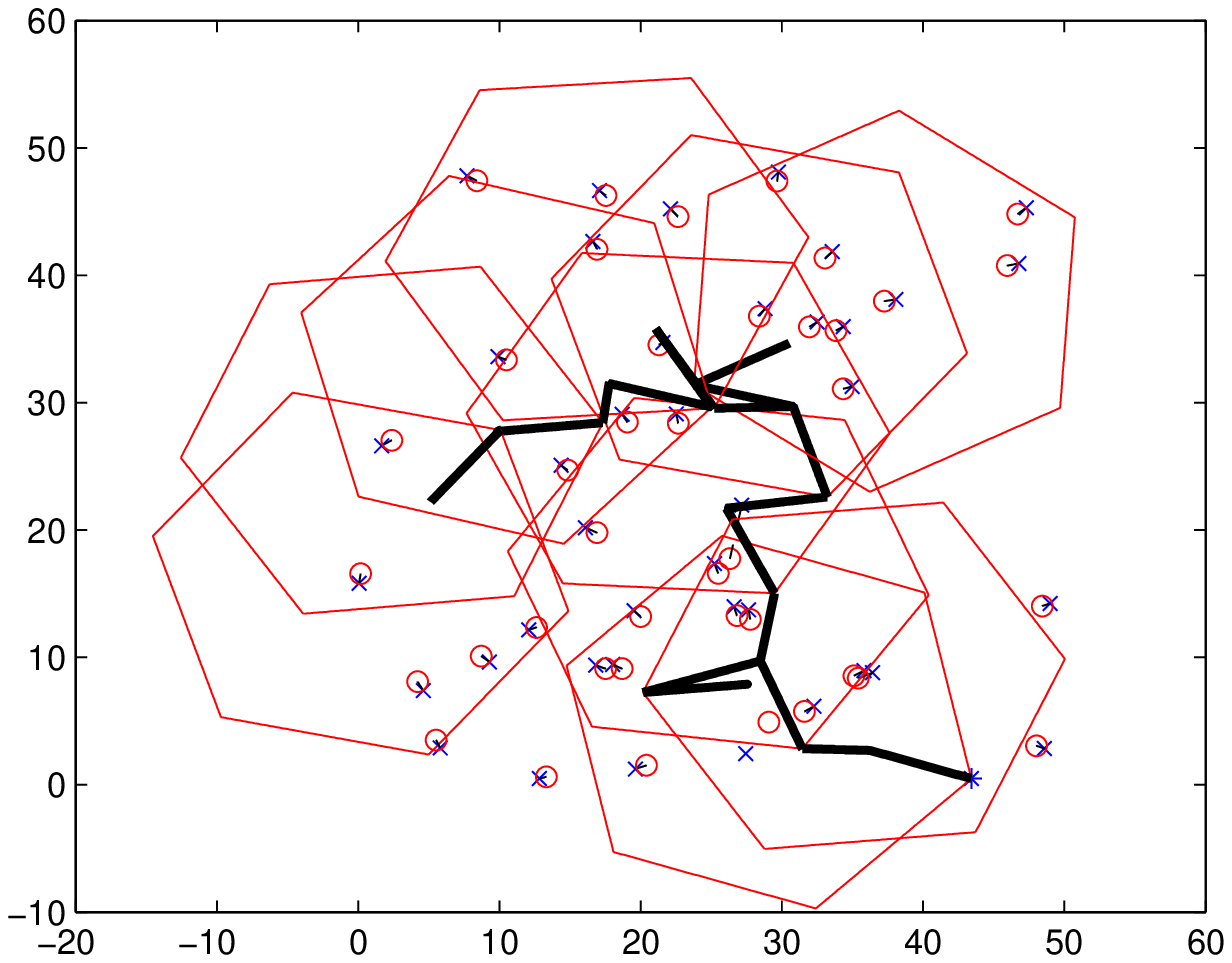}}
\caption{Hexagonal movement pattern in a connected network}
\label{f:fig11}
\end{figure}
\begin{figure}[h]
\centering
\subfloat[$n=75$, $r=12$, $D=1158$, Average error=0.92]{\label{f:fig12a} \includegraphics[width=0.5\textwidth]{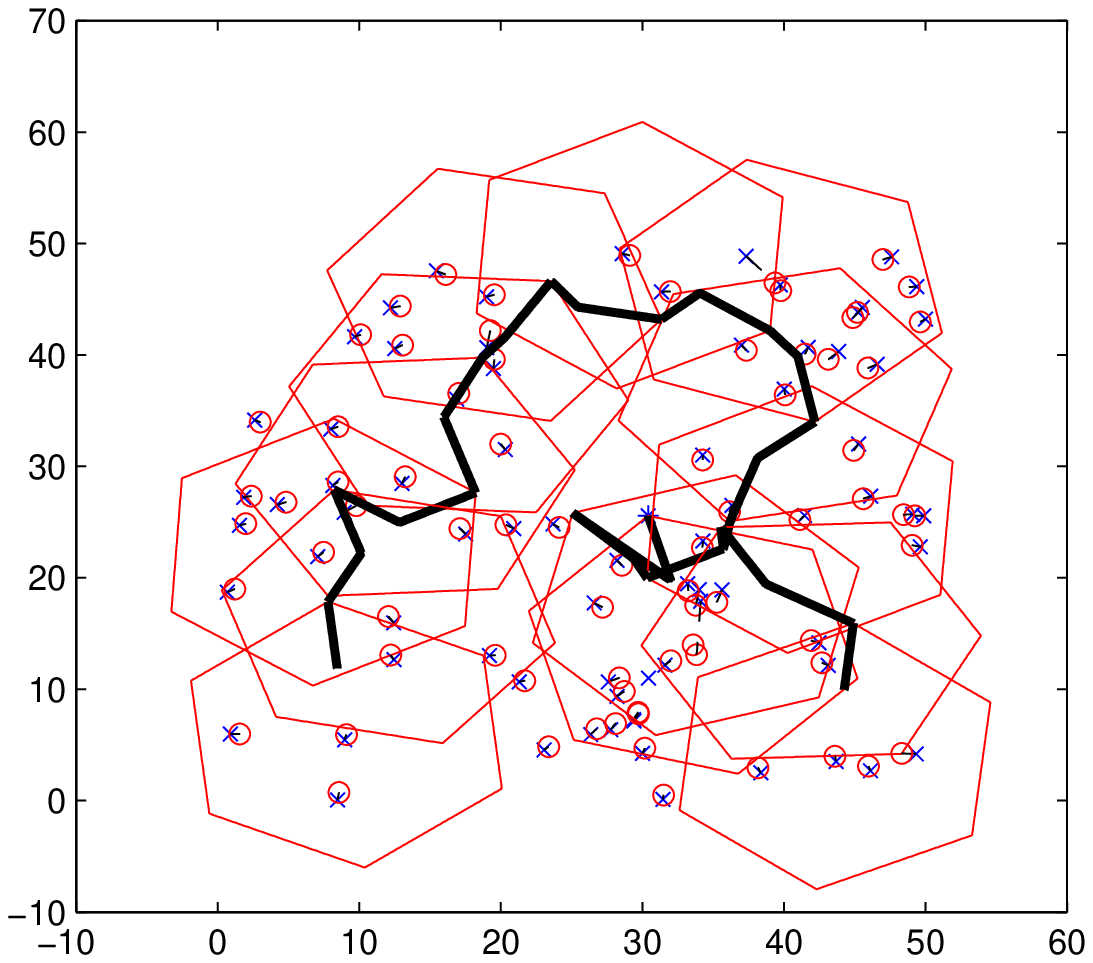}}
\subfloat[$n=100$, $r=10$, $D=1481$, Average error=0.76]{\label{f:fig12b} \includegraphics[width=0.53\textwidth]{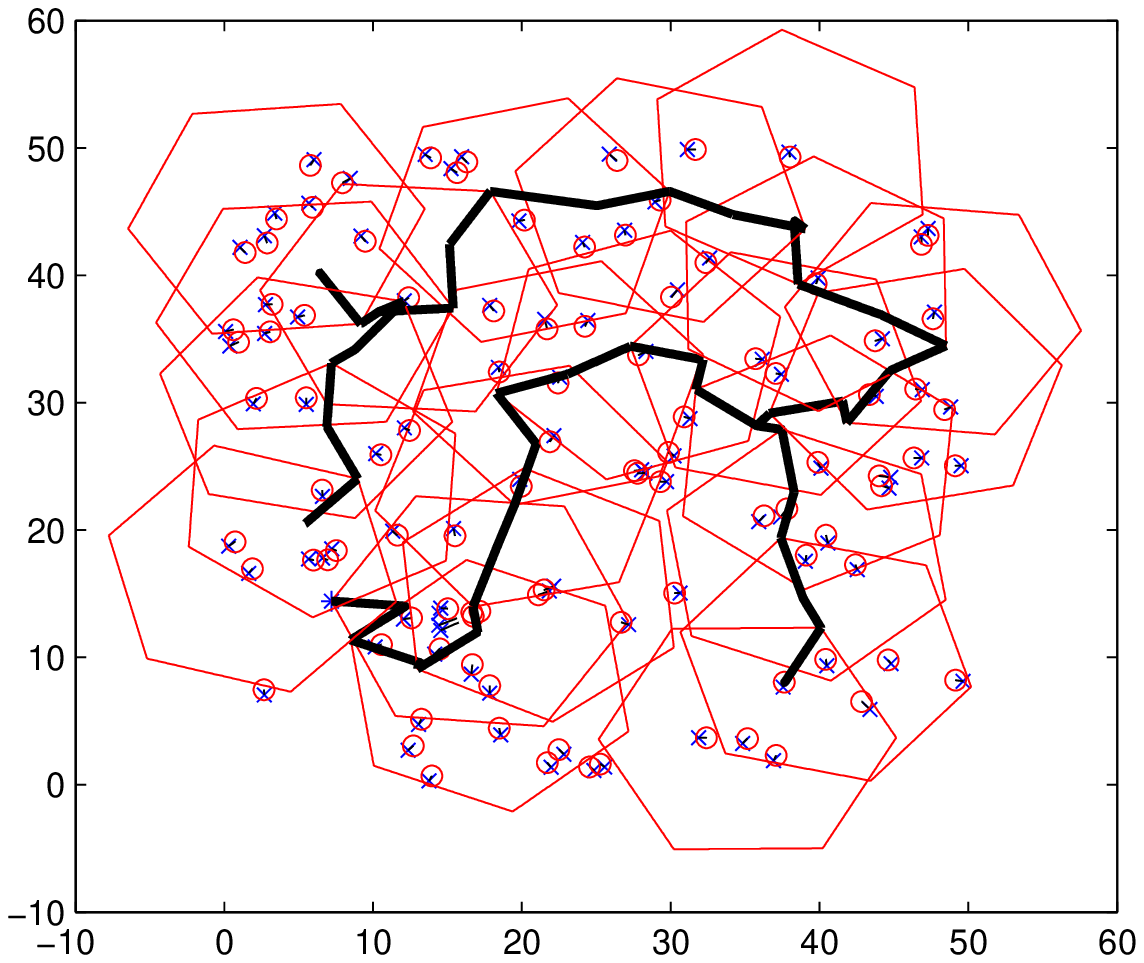}}
\caption{Hexagonal movement pattern in a connected network}
\label{f:fig12}
\end{figure}
\begin{table}[]
\centering
\caption{Showing average error (in meter) for different communication range and beacon distance \label{table:5}}
\begin{tabular}
{p{4.5cm}|p{1.2cm}|p{1.2cm}|p{1.2cm}|p{1.2cm}|p{1.2cm} }
\hline
\hline
Beacon distance $\rightarrow$  & $r/10$ &$r/15$  &$r/20$  &$r/25$  &$r/30$  \\
Communication range($\downarrow$)  &&&&&\\
\hline
10 &1.47  & 1.10 & 0.61 & 0.42 & 0.33  \\
15 &2.35  & 1.36 & 1.01 & 0.73 & 0.54 \\
20 &2.74  & 1.90 & 1.28 & 0.93 & 0.62 \\
25 &4.35  & 2.38 & 1.53 & 1.22 & 0.94 \\
30 &5.89  & 3.14 & 2.12 & 1.64 & 1.13 \\
\hline
\hline
\end{tabular}
\end{table}

In Table \ref{table:5}, we have shown localization error for different communication ranges $r$ and beacon distances $u$.
According to Table \ref{table:5}, localization error decreases as beacon distance decreases for any fixed communication range.
Here, $u=r/k$, where $k>7.5$ according to Theorem \ref{the1:errbd} and Theorem  \ref{the2:localz}.
We have also shown path lengths of the mobile anchor for different number of sensors forming a connected network in a square region of fixed side length 50 meter in  Table \ref{table:2}. Here we keep the communication range fixed at 10 meter. Results show that path length does not increase much with number of sensors. It happens since degree of connectivity increases with the number of sensors in a fixed region.

Next we show the simulation results for the path planning scheme for a bounded square region.
We have simulated the proposed scheme for a bounded region by varying $X$ from $r/10$ to $r/20$ and compared with
\cite{Chia-Ho-Ou2013}. We have fixed the value of communication range $r=10$ meter. We have compared path lengths with the existing schemes
by varying $X$ for a fixed square region with side length $L=200$ meter. Results are shown in Table \ref{table:3}.
\begin{table}[]
\center
\caption{Average path length (in meter) of our scheme varying number of sensors}
\label{table:2}
\begin{tabular}{p{4cm}|p{1.5cm}|p{1.5cm}|p{1.5cm}|p{1.5cm}|p{1.5cm}}
\hline
\hline
   No. of sensors  & 100   & 150     & 200  & 250     & 300  \\
  &      &     &  & &      \\
\hline
 Average path length (meter) & 1490 & 1550 & 1640 & 1675 & 1754   \\
\hline
\hline
\end{tabular}
\end{table}
\begin{table}[h]
\center
\caption{Comparison of path length (in meter) of our scheme with existing schemes varying $X$}
\label{table:3}
\begin{tabular}
{p{1cm}|p{1.4cm}|p{2.2cm}|p{1.6cm} |p{1.4cm}|p{1.4cm}|p{1.4cm}}
\hline
\hline
X($\downarrow$) & $D_{Hexagon}$  & $D_{Chia-Ho-Ou}$  & $D_{Doublescan}$  & $D_{Hilbert}$  & $D_{Circles}$ & $D_{S-curves}$ \\
\hline
$r/10$ & 4987 & 5945 & 6019 & 5377 & 6228 & 5431 \\
$r/15$ & 4292 & 5724 & 5827 & 5185 & 6013 & 5124 \\
$r/20$ & 4271 & 5728 & 5735 & 5094 & 5911 & 5420 \\
\hline
\hline
\end{tabular}
\end{table}
\begin{table}[h]
\center
\caption{Percentage (\%) of improvement of our scheme in terms of path length compared to existing schemes for communication range $r=10$}
\label{table:6}
\begin{tabular}{p{5.4cm}|p{2.2cm}|p{1.7cm}|p{1.2cm}|p{1.2cm}|p{1.2cm}}
\hline
\hline
Existing schemes & $D_{Chia-Ho-Ou}$  & $D_{Doublescan}$  & $D_{Hilbert}$  & $D_{Circles}$ & $D_{S-curves}$ \\
\hline
\% of improvement of our scheme &      &     &  & &      \\
compared to the existing schemes in terms of path length & 22.12\% & 22.93\% & 13.45\% & 25.35\% & 15.19\% \\
\hline
\hline
\end{tabular}
\end{table}
We have shown the percentages of improvement of our scheme over various schemes compared to path length for $r=10$ in Table \ref{table:6}. We have taken averages for different $X$ values to find average improvement on path length. Results show $13.45 \%$ to $25.35 \%$ improvement over different schemes. We have also shown values of average localization error of our scheme in Table \ref{table:4} varying communication ranges $r$ and beacon distance $u$. We have taken $X=u$ in Table \ref{table:4}. Here also, localization error decreases as beacon distance decreases for any fixed communication range.

\begin{table}[]
\centering
\caption{Showing average error (in meter) for different communication range and beacon distance \label{table:4}}
\begin{tabular}
{p{4.5cm}|p{1.2cm}|p{1.2cm}|p{1.2cm}|p{1.2cm}|p{1.2cm} }
\hline
\hline
Beacon distance $\rightarrow$  & $r/10$ &$r/15$  &$r/20$  &$r/25$  &$r/30$  \\
Communication range($\downarrow$)  &&&&&\\
\hline
10 &1.59  &1.17  &0.78  &0.57  &0.46  \\
15 &2.75  &1.58  &1.22  &0.98  &0.85  \\
20 &2.93  &1.95  &1.61  &1.17  &0.89  \\
25 &4.28  &2.47  &1.61  &1.19  &0.99  \\
30 &4.94  &2.65  &1.72  &1.36  &1.11  \\
\hline
\hline
\end{tabular}
\end{table}

\section{Conclusion}
In this paper we have proposed path planning for a mobile anchor in a connected network and also in a bounded rectangular region.
Our movement strategy reduces the requirement of three beacon points for localization to two beacon points, which helps improving
path length for a rectangular region.
In a connected network, once a sensor is localized, our path planning is able to localize all its neighbors
with one hexagonal movement around the sensor. After completing one hexagonal movement, anchor decides its next
destination depending upon received information from the neighboring sensors and localize all sensors along
the way it moves. The novelty is that without knowing the boundary of the network,
our distributed algorithm localizes all sensors using connectivity without any range estimation.
We have also computed the length of the path traversed by the anchor for different number of sensors with good
localization accuracy based on simulation.
We have compared path length of our proposed movement strategy for rectangular region theoretically with existing
literature to show better performance. Simulation results show $13.45 \%$ to $25.35 \%$ improvement of our path
planning strategy over existing strategies in terms of path length of mobile anchor for rectangular region.
In future we will try to investigate path planning in presence of obstacles.
\label{sec:conclusion}

\section*{Acknowledgement}
The first author is thankful to the Council of Scientific and Industrial Research (CSIR), Govt. of India,
for financial support during this work.

\bibliographystyle{plain}
\bibliography{ref}
\end{document}